\documentclass[journal, onecolumn, 12pt]{IEEEtran}
\usepackage[xx,debug]{optional}        
\usepackage{times}
\usepackage{amsmath}
\usepackage{amsthm}
\usepackage{subfigure}
\usepackage{xcolor}
\usepackage{verbatim}
\usepackage{epsfig}
\usepackage{amsfonts}
\usepackage{latexsym}
\usepackage{delarray}
\usepackage[boxed, linesnumbered]{algorithm2e}
\usepackage{balance}

\usepackage{cite}
\graphicspath{{./figs/}}
\usepackage{array}
\usepackage{fixltx2e}

\usepackage{stfloats}
\def\minim{\wedge}
\def\maxim{\vee}

\def\N{{\cal B}}

\def\conv{\otimes}

\def\maxconv{\overline{\otimes}}

\def\maxconv{\,\overline{\otimes}\,}

\def\N{{\cal N}}

\def\Dl{{d\ell}}  

\def\F{{\mathcal{F}}}
\def\T{{\mathcal{T}}}

\def\RR{{\mathbb{R}}}

\def\ZZ{{\mathbb{Z}}}

\newtheorem{theorem}{Theorem}

\newtheorem{lemma}{Lemma}
\newtheorem{corollary}{Corollary}

\begin{document}

\title{A Fluid-Flow Interpretation of SCED Scheduling }

\author{\IEEEauthorblockN{J\"{o}rg Liebeherr}\\
\IEEEauthorblockA{Department of Electrical and Computer Engineering\\
University of Toronto}
}



\maketitle

\begin{abstract}
We show that a fluid-flow interpretation of Service Curve Earliest Deadline First (SCED) 
scheduling  simplifies deadline derivations for this scheduler. 
By exploiting the recently reported isomorphism between min-plus and max-plus 
network calculus and expressing deadlines in a max-plus algebra, 
deadline computations no longer require explicit pseudo-inverse computations. 
SCED deadlines are provided for latency-rate as well as a class of piecewise linear 
service curves.

\end{abstract}


\section{Introduction}
\label{sec:intro}

Service Curve Earliest Deadline First (SCED) \cite{SCED} 
offers an alternative viewpoint on  the design of packet scheduling algorithms. 
The usual approach is to first design a scheduling algorithm and then study or analyze its properties.  
SCED proceeds in the reverse order in that it provides a mechanism to realize a scheduling algorithm with given properties. 
The properties, such as guarantees on rate or delays, are expressed in terms of the concept of {\it service curves} of the network calculus \cite{Cruz95}. 
Given a service curve, SCED computes deadlines for arriving traffic 
and transmits traffic in the order of deadlines.   
As long as no deadline is violated,  the scheduler is guaranteed to satisfy 
the service curve guarantees.  The SCED framework in \cite{SCED} is completed by {\it schedulability conditions} that    
predict whether given service curves can 
be met at a transmission link with bounded (not necessarily fixed) capacity.

Since service curves are arbitrary non-negative  
increasing functions, SCED has a great deal of flexibility for  
offering different service guarantees to traffic flows.  For example, it does not share the well-known  
drawback of weighted fair scheduling algorithms \cite{WFQ-keshav} when  
providing low delays to low-bandwidth traffic. 
(To achieve low delays, fair schedulers must 
increase the guaranteed rate of a flow). Guarantees in SCED provide lower bounds on the service. SCED+ \cite{SCED+} is an extension to guarantee bounds on the delay jitter.  
Scheduling algorithms inspired by SCED, such as Hierarchical Fair Service Curve (HFSC) \cite{HFSC} are widely deployed in the network stack of current operating systems \cite{ALTQ,TC-Linux}. 

Even though SCED appears an ideal vehicle for studying and interpreting scheduling algorithms with  methods of the network calculus, it has not played a major role in recent network calculus research. A closer inspection of SCED provides clues that may offer an explanation for the lack of interest.  First, the original formulation of SCED in \cite{SCED} assumes that all packets have the same  size.
The condition on equal packet sizes is relaxed in \cite[Sec. 2.3.2]{Book-LeBoudec} 
by adding an additional `packetizer' service element. The setup of the analysis in \cite{SCED}  also requires that a packet can 
depart in the same time slot where it arrives. This corresponds to an assumption of `cut-through' switching in a network, whereas most networks perform `store-and-forward' switching. Finally, since SCED operations are described 
within the framework of the min-plus calculus, 
deadlines are expressed as a pseudo-inverse of a traffic function, 
which is not very intuitive. 

In this paper, we will show that the above issues can be resolved when describing SCED in terms of the max-plus network calculus. 
By adopting a fluid-flow interpretation of SCED operations, we 
can exploit the recently established duality between min-plus and max-plus network calculus \cite{Duality} for an analysis of SCED. 
Since max-plus 
expressions are more convenient for computing timestamps, we use them 
for  deadline computations  in SCED. For schedulability conditions, we resort to min-plus expressions, since the corresponding 
max-plus conditions become unwieldy. 

In Sec.~\ref{sec:netcalc} we briefly discuss 
network calculus concepts used in this paper. In Sec.~\ref{sec:fluidSCED} we discuss SCED operations in terms of max-plus algebra expressions. In Sec.~\ref{sec:sched}
we derive schedulability conditions for the fluid-flow SCED scheduler. 
In Sec.~\ref{sec:deadline} we address the computation of SCED deadlines in fluid-flow SCED. In Sec.~\ref{sec:packet} we address deadline computations in a  
packet-level system.

\section{Duality of Min-Plus and Max-Plus Network Calculus}
\label{sec:netcalc}

The continuous-time min-plus network calculus conducts an analysis of network elements within a 
$(\F_o, \minim, \conv)$ dioid algebra, where $\F_o$ is the set of 
left-continuous, non-decreasing 
functions $F: \RR \to \RR^+_o \cup \{+\infty\}$, with  $F (t) = 0$ if $t \le 0$, 
the $\minim$-operation is a pointwise minimum, and $\conv$ is the min-plus convolution, which 
is defined as $F \conv G (t) = 
\inf_{0 \leq s \leq t} \left\{ F(s) + G(t-s) \right\}$ for two functions $F, G \in \F_o$.
The cumulative amount of arrivals and  departures at a network element in the time interval $[0,t)$ is given by $A(t)$ and $D(t)$, respectively, with 
$A, D\in \F_o$.
The available service at a network element is expressed in terms of a function $S \in \F_o$, referred to as {\it minimum service curve}, 
which satisfies $D (t) \ge A \conv S(t)$ for all $t$. 
When arrivals are bounded by a function $E \in \F_o$, such that $E (s) \ge A (t+s) - A(t)$ for all $s$ and $t$, we say that~$E$ is a {\it traffic envelope} for $A$. 

Functions in the  max-plus network calculus compute the time of an arrival or departure event 
for a given number of bits. The continuous-space version uses a $(\T_o, \maxim, \maxconv)$ dioid, 
where $\T_o$ is the set of 
right-continuous, non-decreasing 
functions \mbox{$F: \RR \to \RR^+_o \cup \{-\infty\} \cup \{+\infty\}$}, with  $F (\nu) = -\infty$ if $\nu < 0$ and $F (\nu) \ge 0$ if $\nu \ge 0$. 
The $\maxim$-operation is a pointwise maximum, and $\maxconv$ is the max-plus convolution, with  $F \maxconv G (\nu) = 
\sup_{0 \leq \kappa \leq \nu} \left\{ F(\kappa) + G(\nu-\kappa) \right\}$ for two functions $F, G \in \F_o$. 
Arrivals and departures are described by functions $T_A \in \T_o$ and $T_D \in \T_o$. 
Here,  $T_A (\nu)$ is the arrival time of bit $\nu$, where bit values are allowed to be real numbers. 
A minimum service curve is a function $\gamma_S \in \T_o$ such that $T_D (\nu) \le T_A \maxconv \gamma_S(\nu)$ for all $\nu$, and a traffic envelope $\lambda_E \in \T_o$ for an arrival 
time function satisfies 
$\lambda_E (\mu) \le T_A (\nu+\mu) - T_A(\nu)$ for all $\nu$ and~$\mu$. 

As shown in \cite{Duality}, there exists an isomorphism between the min-plus and max-plus network 
calculus via the  pseudo-inverse functions 
\begin{eqnarray*}
F^{\,\downarrow} (y) &  = \ \inf \left\{  x \mid  F(x) \ge y \right\}    & = \ \sup \left\{  x \mid  F(x) < y \right\} \,  , 
\label{max-plus-survey:eq-define-lower-pseudo-inverse} \\
F^{\,\uparrow} (y)  & = \ \sup \left\{ x  \mid  F(x) \le y \right\}   & = \ \inf \left\{ x  \mid  F(x) > y \right\} \, ,  
\end{eqnarray*}
where $F^{\,\downarrow}$ is referred to as lower pseudo-inverse and  $F^{\,\uparrow}$ as
upper pseudo-inverse. The pseudo-inverses establish the following relationships: 
\begin{itemize} 
\item  $F \in {\cal F}_o $ $\Rightarrow$  $F^{\,\uparrow} \in {\cal T}_o$.
\item $F \in {\cal T}_o $ $\Rightarrow$   $F^{\,\downarrow} \in {\cal F}_o$ 
\item $F^{\,\downarrow}$ is left-continuous and $F^{\,\uparrow}$ is right-continuous.
\item $F$ is left-continuous $\Rightarrow$
$ F =   \bigl(F^{\,\uparrow}  \bigr)^{\,\downarrow}$.
\item $F$ is right-continuous $\Rightarrow$
$F =   \bigl(F^{\,\downarrow}  \bigr)^{\,\uparrow}$.
\end{itemize}
With the pseudo-inverses, we can map operations between the min-plus and max-plus network 
calculus by 
\begin{itemize} 
\item \( 
\bigl(F \minim  G \bigr)^{\,\uparrow} (\nu)= 
F^{\,\uparrow} \maxim G^{\,\uparrow}  (\nu)
\).

\item 
\(
\bigl(F \conv G \bigr)^{\,\uparrow} (\nu) =  
F^{\,\uparrow}\maxconv G^{\,\uparrow}  (\nu)
\).


\item 
\( \bigl(F + G \bigr)^{\,\uparrow} (\nu)   
=  \displaystyle \inf_{0 \le \kappa \le \nu} \max \bigl\{F^{\,\uparrow} (\kappa), G^{\,\uparrow}(\nu-\kappa)\bigr\}  
\).
\end{itemize}
For mapping in the other direction we have 
\begin{itemize} 
\item 
\(
 \bigl(F \maxim  G \bigr)^{\,\downarrow} (t) = 
 F^{\,\downarrow} \minim G^{\,\downarrow}  (t) \, . 
\)

\item 
\(
\bigl(F \maxconv G \bigr)^{\,\downarrow} (t)  =  
F^{\,\downarrow}\conv G^{\,\downarrow}  (t) \, . 
\)


\item 
\(    
  \displaystyle \Bigl( \inf_{0 \le s \le t} \max \bigl\{F (s), G(t-s)\bigr\} \Bigr)^{\,\downarrow} = F^{\,\downarrow} (t) + G^{\,\downarrow}  (t) \, . 
\)

\end{itemize} 
With this, we can set 
$A \equiv T_A^\downarrow$ and $D \equiv T_D^\downarrow$, as well as $T_A\equiv A^\uparrow$ and $T_D \equiv D^\uparrow$.
Service curves and traffic envelopes are related as follows:
\begin{itemize} 
\item 
\(
D (t) \ge A \conv S(t) \, , \forall t  \Rightarrow T_D (\nu) \le T_A \maxconv S^\uparrow (\nu) \, , \forall \nu
\). 

\item 
\(
E (s) \ge A (t+s) - A(t) \, , \forall t , s  \Rightarrow E^\uparrow (\mu) \le T_A (\nu+\mu) - T_A(\nu)
 \, , \forall \nu , \mu
\). 

\item 
\(
T_D (\nu) \le T_A \maxconv \gamma_S (\nu)   \, , \forall \nu\Rightarrow D (t) \ge A \conv \gamma_S^\downarrow (t) \, , \forall t  
\). 

\item 
\(\lambda_E (\mu) \le T_A (\nu+\mu) - T_A(\nu)
 \, , \forall \nu , \mu \Rightarrow \lambda_E^\downarrow (s) \ge A (t+s) - A(t) \, , \forall t , s   
\). 
\end{itemize} 
With our convention to use $S$ and $E$ for service curves and traffic envelopes in the min-plus network 
calculus, and $\gamma_S$ and $\lambda_E$ in the max-plus network calculus, we can set $S \equiv \gamma_S^\downarrow$ and $\gamma_S \equiv S^\uparrow$, as well as 
$E \equiv \lambda_E^\downarrow$ and $\lambda_E \equiv E^\uparrow$. 

As argued in \cite{Duality}, there is no isomorphism when the min-plus network calculus 
is defined in discrete~time ($t \in \ZZ$) or the max-plus calculus is defined in  
discrete~space ($\nu \in \ZZ$). It also does not exist for a packet-level characterization of traffic.  Hence, to exploit the above relationships within SCED, we must resort to a 
fluid-flow description of traffic, where time and space are expressed by non-negative real numbers.  

\section{Fluid-flow SCED}
\label{sec:fluidSCED}

The objective of SCED is  a scheduling mechanism that can realize any given minimum service curve.  The basic idea is to assign arriving traffic a deadline equal to the 
latest departure time permitted by the given service curve. 
As long as all traffic departs before the expiration of the assigned deadlines, the service curve is guaranteed to hold. 
\begin{figure}[!t]
\centering 
\subfigure[Space-domain view.]{\includegraphics[width=0.49\columnwidth]{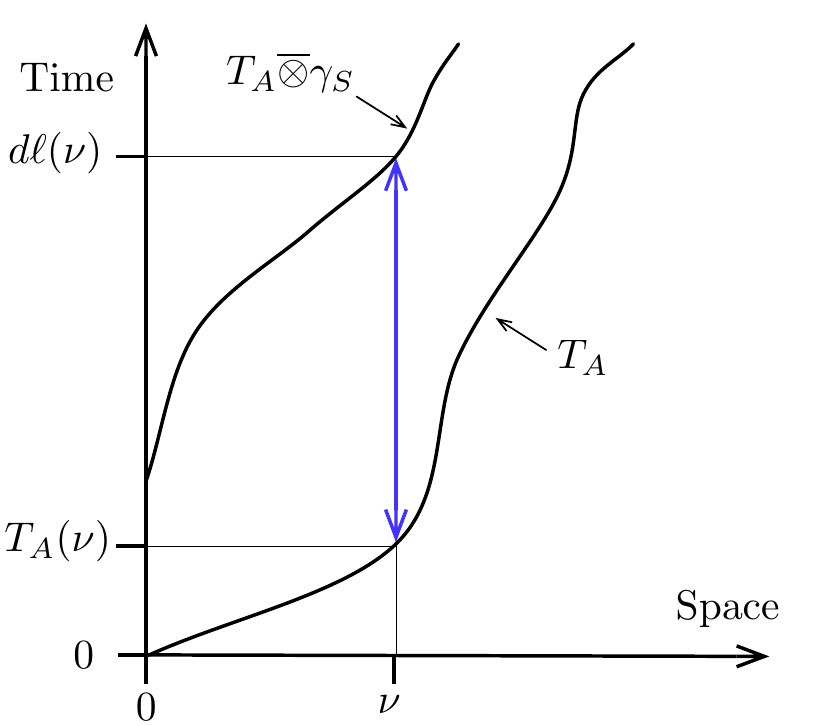}\label{fig:SCED-spacedomain}}%
%
\subfigure[Time-domain view.]{\includegraphics[width=0.49\columnwidth]{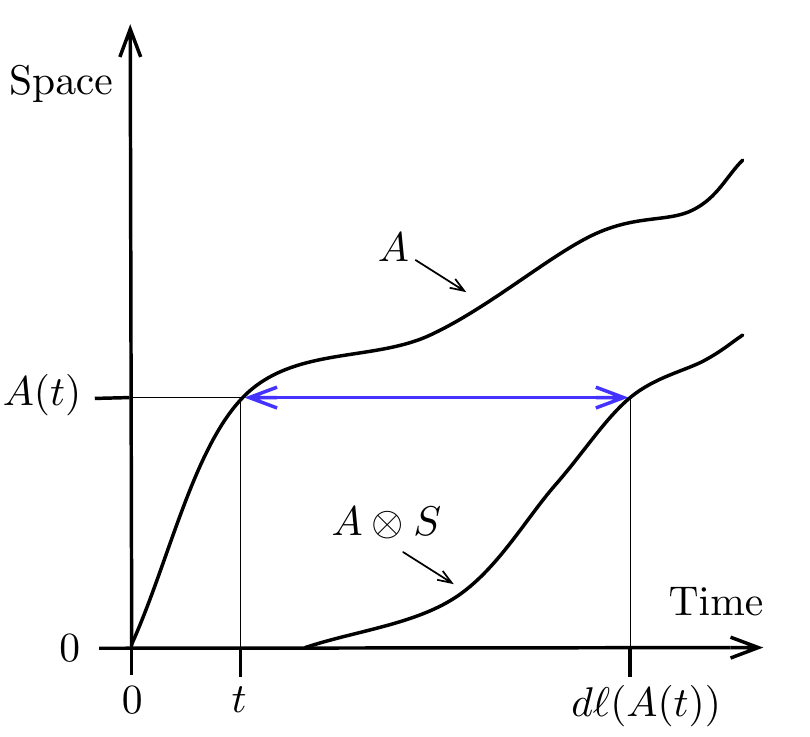}\label{fig:SCED-timedomain}}
\caption{SCED Deadlines.}
\label{fig:SCED-deadlines}
\vspace{-10pt}
\end{figure}


We first discuss the deadline assignment from the perspective of the max-plus algebra.  
We consider a fluid-flow version of SCED, where each  bit value $\nu \in \RR_o^+$ is assigned a deadline $\Dl(\nu)$. All traffic is transmitted in the order of deadlines. The deadline assignment is illustrated in 
Fig.~\ref{fig:SCED-spacedomain}. Bit $\nu$ with arrival time $T_A(\nu)$ is assigned 
the deadline 
\begin{align}
\Dl (\nu) = T_A \maxconv \gamma_S (\nu) \, , 
\label{eq:Deadline-assign}
\end{align}
where $\gamma_S$ is a max-plus minimum service curve. This gives  
the equivalency
\begin{align}
T_D (\nu)\le \Dl (\nu) 
\qquad \Longleftrightarrow \qquad
T_D (\nu) \le T_A \maxconv \gamma_S (\nu) \, . 
\label{eq:Deadline-scerv-equiv}
\end{align}
Hence, if all traffic departs 
by its deadline, $\gamma_S$ is a minimum service curve. 
Conversely, if $\gamma_S$ is a minimum service curve, then there is 
no deadline violation. 

The deadline assignment is more intricate when we describe it in terms of min-plus network calculus expressions. The deadline assignment is sketched in Fig. ~\ref{fig:SCED-timedomain} for a continuous arrival time function $A$. 
The deadline of an arrival just before time $t$,\footnote{Note that 
$A(t)$ does not include  arrivals that occur at time $t$.} denoted 
by $\Dl (A(t))$,  is set to the time after $t$ when $A \conv S$ 
has caught up to $A(t)$. If the departures at time $\Dl (A(t))$, given by $D (\Dl (A(t)))$, are at least $A \conv S (\Dl (A(t)))$, then $S$
satisfies the service curve requirement $D (\Dl (A(t))) \ge A \conv S (\Dl (A(t)))$. 
The computation of the deadline involves the computation of an inverse. More precisely, since neither $A$ nor $A \conv S$ are continuous or strictly increasing, the deadline requires to take a pseudo-inverse. 
By choosing the upper pseudo-inverse, we  recover the deadline 
from~\eqref{eq:Deadline-assign}, since 
\begin{align*}
\Dl (\nu) & = T_A \maxconv \gamma_S (\nu) \\
& = A^{\uparrow} \maxconv S^{\uparrow} (\nu)\\
& = \left( A \conv S \right)^\uparrow (\nu) \\
& = \sup \bigl\{  \tau \mid A \conv S (\tau) \le \nu \bigr\} \, . 
\end{align*}
Then, the deadline for $A(t)$  is given by 
\[
\Dl (A(t))  = \sup \bigl\{  \tau \mid A \conv S (\tau) \le A(t) \bigr\} \, . 
\]
Note that the computation of the pseudo-inverse  
for an arrival time $t$ requires to compute $A \conv S (\tau)$ for values $\tau >t$, which appears to assume knowledge of future arrivals (after time $t$). Fortunately, this is not the case, 
since for $\tau > t$, 
\[
 A \conv S (\tau) = \inf_{0 \leq s \leq \tau} \left\{ A(s) + S(\tau-s) \right\} \le A (t) 
\]
if and only if
\[
\inf_{0 \leq s \leq t} \left\{ A(s) + S(\tau-s) \right\} \le A (t) \, . 
\]
Despite the additional complexity of deadline computations in a min-plus setting, all discussions of SCED in the literature 
\cite{SCED,SCED+,Book-Chang,Book-LeBoudec,SariowanPhD} have chosen a min-plus formulation. 
Interestingly, the computations in these works use  
the lower pseudo-inverse for the computation of deadlines.

\section{Schedulability Condition of Fluid-Flow SCED}
\label{sec:sched}

In this section we derive a schedulability condition that determines whether a SCED scheduler at a link with variable transmission rate 
can support a set of min-plus or max-plus service curves for a set of flows.  
Our derivations will use the min-plus network calculus, since the max-plus version of the schedulability condition is generally  not useful for practical computations 
(as shown below). 

We consider a set $\cal N$ of flows. Let $A_j$ and $D_j$ denote the time-domain arrival and departure functions of flow $j \in {\cal N}$. The functions $T_{A_j}=A_j^\downarrow$ and $T_{D_j} = D_j^\downarrow$
denote the space-domain formulations of arrivals and departures. Arrivals of flow 
$j$ are constrained by a traffic envelope, which is denoted by either $E_j$ or 
$\lambda_j = E_j^\downarrow$. 

We consider a work-conserving link with a time-variable transmission rate.  
We assume that the transmissions of the link can be bounded by a strict service curve $C \in \F_o$ \cite{Book-LeBoudec}, defined by the property that 
for any time interval $(s, t]$ with positive backlog, 
\[
\sum_j \bigl( D_j(t) - D_j(s) \bigr) \ge C (t -s) \, . 
\]
With a constant-rate link, we have $C(t) = c t$ for some  $c >0$.
In the case of packet-level traffic, the transmission of a packet is never 
interrupted, even if a packet arrives with a shorter deadline than 
the packet in transmission. This is referred to as {\it non-preemptive} scheduling. 
In contrast, with {\it preemptive} scheduling, the link always transmits traffic 
with the earliest deadline. 

We are interested in  deriving a  condition that  can determine whether a SCED scheduler  is able to guarantee service curves 
$S_j$ or $\gamma_{S_j}=S_j^\uparrow$ for each flow $j \in {\cal N}$. 
The deadline assignment for each $\nu \ge 0$ is such that 
\begin{align}
\Dl_j (\nu) = T_{A_j} \maxconv \gamma_{S_j} (\nu) = (A_j \conv S_j)^\uparrow (\nu) \, . 
\label{eq:SCED-flow-deadline}
\end{align}

\subsection{Preliminary Results}
We first present preliminary results that will aid in the derivation of the schedulability condition. 
We define $A_j^{<t} (\tau)$ as 
\[
A_j^{<t} (\tau) = \sup \bigl\{ \nu \mid 0 \le T_{A_j} (\nu) < \tau 
\ \text{ and } \ 
{\Dl}_j (\nu) < t \bigr\} \, , 
\]
which are the  arrivals from flow $j$ in the time interval $[0,\tau)$ 
with a deadline less than $t$. We will use the short hand 
$A_j^{<t} (s, t ) = A_j^{<t} (t) - A_j^{<t} (s)$. 
Now, let $t^*$ be the last time before $t$ ($t^* \le t$) when 
the link does not have any backlog from traffic with a deadline before time $t$. 
Also, let $\ell (t^*) \ge 0$ be the untransmitted portion of the packet that is in transmission at time $t^*$. This packet has a deadline 
greater than or equal to~$t$. 
We next relate the function 
$A_j^{<t}$ to deadline violations in SCED. 
\begin{lemma}
\label{lemma:violation}
If a non-preemptive SCED scheduler experiences a deadline violation by time $t$, 
then 
\[
\sum_{j \in {\cal N}} A_j^{<t} (t^*, t)  
+ \ell (t^*) > C (t - t^*) \, . 
\]
\end{lemma}
\begin{proof}
Let us first ignore that packet transmissions cannot be preempted. 
If we have a deadline violation by time $t$,  the 
amount of traffic with a deadline before $t$ exceeds the transmission 
capacity of the link. Since there is no traffic at the link at time 
$t^*$ with a deadline before $t$, we can ignore all arrivals 
and transmissions before $t^*$. 
Then, the arrivals from flow $j$ in $[t^*, t)$ with a deadline 
before $t$ is given by $A_j^{<t} (t^*, t)$. The least available transmission 
capacity of the link  in $[t^*, t)$ is given by $C (t - t^*)$. Therefore, a deadline violation by~$t$ implies that 
$\sum_{j \in {\cal N}} A_j^{<t} (t^*, t)  
 > C (t - t^*)$. Without packet preemption, the remaining part of 
 the packet in transmission at time $t^*$, $\ell (t^*)$,  is added to the workload 
 that must be transmitted before $t$, which yields the claim. 
\end{proof}

The next lemma provides an interesting property of the function $A_j^{<t}$. 
\begin{lemma}\label{lemma:arrivals}
For all $t \ge 0$, we have $A_j^{<t} (t) = A_j \conv S_j (t)$. 
\end{lemma}
\begin{proof}
Setting $t = \tau$ in the definition of $A_j^{<t} (\tau)$, we get 
\[
A_j^{<t} (t) = \sup \bigl\{ \nu \mid  T_{A_j} \maxconv \gamma_{S_j} (\nu) < t \bigr\} \, , 
\]
since $F \maxconv G (\nu) \ge F (\nu)$ for $F, G \in \T_o$. Writing deadlines in 
terms of the min-plus algebra, we obtain 
\begin{align*}
A_j^{<t} (t) & = \sup \bigl\{ \nu \mid  (A_j \conv S_j)^\uparrow (\nu) < t \bigr\} \\
& = \bigl( (A_j \conv S_j)^\uparrow \bigr)^\downarrow (t) \\
& = A_j \conv S_j (t) \, , 
\end{align*}
where the second line uses the lower-pseudo inverse, and 
the last line follows from $ F = ( F^\uparrow)^\downarrow$ if $F \in \F_o$. 
\end{proof}

\subsection{Main Result}

In this section we prove schedulability 
conditions for the fluid-flow SCED scheduler.

\begin{theorem} \label{theorem:SCED-sched}
A non-preemptive SCED scheduler with a set~$\cal N$ of flows as discussed at the 
beginning of this section guarantees the service curves $\{S_j\}_{j \in {\cal N}}$ 
if for all $t \ge 0$ 
\begin{align}
\sum_{j \in {\cal N}} E_j \conv S_j (t) \le \left[ C (t) - \ell_{\max} \right]^+\, , 
\label{eq:SCED-sched-cond}
\end{align}
where $\ell_{\max}$ is the maximum packet size and $[x]^+ = \max \{ x , 0\}$.
\end{theorem}
Note that the condition 
requires that $S_j (t) = 0$ for $t \le C^{\uparrow} (\ell_{\max})$.   
This can be ensured by `appending' a delay element with 
service curve $\delta_{C^{\uparrow} (\ell_{\max})}$ to a given service curve $\gamma$ 
via $\gamma \conv \delta_{C^{\uparrow} (\ell_{\max})}$. 

\medskip
With preemptive scheduling we set $\ell_{\max}=0$. 
We point out that the schedulability condition of preemptive 
SCED when expressed in the max-plus algebra \cite[see Corollary~12.5]{Duality} is 
\[ 
\inf_{\substack{\nu_1, \ldots , \nu_N\\ \nu = \nu_1 + \ldots + \nu_N}}\max_{j = 1, \ldots, N}  E_j^\uparrow \maxconv S_j^\uparrow (\nu_j)  \ge \frac{\nu}{C}\, , \quad \forall \nu \ge 0 \,  . 
\] 
Clearly, this condition is not useful for practical schedulability tests. 
\begin{proof}
According to~\eqref{eq:Deadline-scerv-equiv}, the deadline assignment from \eqref{eq:SCED-flow-deadline} does not 
result in a deadline violation if and only 
if  $\gamma_{S_j} = S_j^\downarrow$ is a max-plus 
service curve for flow $j$ ($T_{D_j} \le T_{A_j} \maxconv \gamma_{S_j}$). 
By the duality properties, $S_j$ is then a min-plus service curve. 
We will show that a deadline violation implies 
that~\eqref{eq:SCED-sched-cond} does not hold. Hence, if~\eqref{eq:SCED-sched-cond} holds, there cannot be a deadline violation. 

Assume that there is a deadline violation before~$t$, and 
let~$t^*$ be as defined above. 
Each flow $ j \in \N$ satisfies
\begin{align}
A_j^{<t} (t^*) = A_j (t^*) 
\label{eq:proof-sched-1}
\end{align}
This a consequence from the fact that earlier arrivals of 
flow~$j$ have an earlier deadline. Since there are arrivals after 
time~$t^*$ with a deadline before $t$, the deadlines 
of all arrivals before~$t^*$ must be less than~$t$. 
We  now derive for flow $j \in \N$
\begin{align*}
A_j^{<t} (t^*, t) & = A_j \conv S_j (t)- A_j (t^*) \\
& = \inf_{0 \leq s \leq t} \left\{ A_j(s) + S_j(t-s) \right\} - A_j (t^*) \\
& \le \inf_{t^* \leq s \leq t} \left\{ A_j(s) + S_j(t-s) \right\} - A_j (t^*) \\
& = \inf_{0\leq s \leq t-t^*} 
\left\{ A_j (t^*+s) - A_j (t^*) + S_j(t-t^*-s) \right\}  \\
& \le \inf_{0\leq s \leq t-t^*} 
\left\{ E_j(s) + S_j(t-t^*-s) \right\} \\
& = E_j \conv S_j (t-t^*) \, . 
\end{align*}
In the first step, we use  Lemma~\ref{lemma:arrivals} and \eqref{eq:proof-sched-1}. The second step simply expands the convolution. The third step 
relaxes the infimum by restricting its range, and the  fourth step makes a 
change of variable. 
The inequality in the fifth step follows since  $E$ is an envelope, 
that is,  $E_j(s) \ge A_j(t^*+s) - A_j (t^*)$, which 
yields the convolution in the last step. 

By Lemma~\ref{lemma:violation}, since 
$E_j \conv  S_j (t) \ge A_j \conv S_j (t)$ and $\ell_{\max} \ge \ell (t^*)$, a deadline violation before $t$ implies that 
\[
\sum_{j \in \N} E_j \conv S_j (t) + \ell_{\max}> C (t)  \, , 
\]
which contradicts~\eqref{eq:SCED-sched-cond}. Thus, we cannot have a deadline violation. Hence, the functions $S_j$ are minimum service curves. 
\end{proof}

The following condition, which follows directly from Theorem~\ref{theorem:SCED-sched}, is useful when no information is available 
on the arrivals. 
\begin{corollary}
Under the assumptions of Theorem~\ref{theorem:SCED-sched}, 
the SCED scheduler guarantees  service curves $\{S_j\}_{j \in {\cal N}}$ if 
for all $t \ge 0$ 
\[
\sum_{j \in {\cal N}} S_j (t) \le \left[ C (t) - \ell_{\max}\right]^+ \, , 
\label{eq:SCED-sched-cond-coro}
\]
\end{corollary}

Since the available transmission capacity of the link in a time interval $(s,t]$ 
may exceed $C (t-s)$, the condition in Theorem~\ref{theorem:SCED-sched} is a sufficient condition. On the other hand, if the link is a fixed-rate work-conserving 
link with exact service curve $C(t) = ct$, such that 
$\sum_j D_j (t) = \sum_j A_j \conv C (t)$ for all  $t \ge 0$, and 
additionally assume that arrivals on a flow may saturate their 
envelopes, that is, $A_j (t) = E_j (t)$,  
we can provide a necessary condition for guaranteeing minimum service curves $\{S_j\}_{j \in {\cal N}}$, which are close to~\eqref{eq:SCED-sched-cond}.

\begin{theorem} \label{theorem:SCED-sched-nec}
Consider a  SCED scheduler that operates at 
a fixed-rate and offers an exact service curve $C(t) = ct$. 
Assume that the arrivals from each flow $j \in \N$ can saturate its 
envelope $E_j$.  
If SCED ensures each flow $j \in {\cal N}$ a minimum service curve $S_j$, 
then,   
for all $t \ge 0$, 
\begin{align}
\sum_{j \in {\cal N}} E_j \conv S_j (t) \le ct   \, . 
\label{eq:SCED-sched-cond-nec}
\end{align}
\end{theorem}
For preemptive scheduling, the condition in~\eqref{eq:SCED-sched-cond-nec} is 
necessary and sufficient. Since, for non-preemptive scheduling,  the condition in 
\eqref{eq:SCED-sched-cond} is not always necessary (e.g., if there is only one flow), reducing the difference between~\eqref{eq:SCED-sched-cond} and~\eqref{eq:SCED-sched-cond-nec} requires knowledge of the number of flows and 
the service curves of each flow. 
\begin{proof}
Suppose that~\eqref{eq:SCED-sched-cond-nec} does not hold for some value of $t$. 
Let the arrivals saturate their envelopes, that is 
$A_j (\tau) = E_j (\tau)$ for all $0 \le \tau \le t$ for each $j \in {\cal N}$. 
Since, by assumption, each of the $\{S_j\}_{j \in {\cal N}}$ is a minimum service curve, we have  for each $j \in {\cal N}$ that 
\[
D_j (t) \ge E_j \conv S_j (t) \, . 
\]
Summing over all flows and using the violation of~\eqref{eq:SCED-sched-cond-nec}, 
we get
\[
\sum_{j \in {\cal N}} D_j (t) \ge \sum_{j \in {\cal N}} 
 E_j \conv S_j (t)  > c t \, . 
\]
However, this is not possible since the aggregate departures from all flows 
in the interval $[0, t]$ cannot exceed $c t$. 
\end{proof}

\section{Computations of SCED Deadlines}
\label{sec:deadline}

As we have seen, 
by avoiding the need to compute  pseudo-inverses, SCED deadlines in a max-plus setting are conceptually simpler and more intuitive than presented in  the literature on SCED \cite{Cruz-ICCCN95,SCED,SCED+,Book-Chang,Book-LeBoudec}. 

{\bf Delay Guarantees:} The max-plus service curve for guaranteeing a delay bound $d$ for traffic is simply $\gamma_S (\nu) = d$. This leads to the deadline computation 
\[
\Dl (\nu) = \max_{0 \le \kappa \le \nu} \{ T_A (\kappa) + d \}=  T_A (\nu) + d \, . 
\]  
That is, the  deadline is the sum of the arrival time and the delay bound. 
This deadline assignment corresponds to that 
of  Earliest-Deadline-First (EDF) scheduling. 

{\bf Rate Guarantees:} The computation of SCED deadlines for a rate guarantee with 
service curve $\gamma_S (\nu) = \tfrac{\nu}{R}$ requires a little bookkeeping at the start of a busy period of a flow. Here, a busy period of a flow is a maximal time interval  where the 
flow has a positive backlog $B$, with $B(t) = A(t) - D(t)$. 
The following computation assumes that all arrivals occur at the start of or within a busy period, and 
that the number of busy periods within any finite time interval is finite. This assumption does not hold for 
general fluid-flow traffic arrivals. In particular, if traffic arrives  at a constant rate and is served at the same rate,  
no backlog builds up, and, hence, there is no busy period. On the other hand, in practical scenarios, 
where arrivals occur in chunks of arbitrary size and the maximum service rate has an upper bound, any arrival 
creates a backlog, and, therefore, starts or falls into a busy period. 

Consider the arrival time $T_A (\nu)$ of a bit value $\nu$. We suppose the arrival 
occurs in a busy period that started at time $\underline{t}$, that is, 
$\underline{t} = \sup \{s \le t  \mid A (t) = D (t) \}$. 
 Let $\underline{\nu}$ be the bit that started the busy period, with arrival time $T_A (\underline{\nu}) = \underline{t}$, 
that is, $\underline{\nu} = \inf \{\kappa \mid T_A (\kappa) \ge \underline{t}\}$. 
For the delay $W$, defined as $W(\nu) = T_D (\nu) - T_A (\nu)$, we have 
\begin{align*}
W (\underline{\nu}) = 0 \, , \ \ W(\kappa) > 0 \, , \ \ \forall \kappa \in (\underline{\nu},  \nu ] \, . 
\end{align*}
(In general, it is possible that  $W (\kappa) > 0$ for all $\kappa \in (\underline{\nu},  \nu ]$ \cite[\S11.5]{Duality}. However, 
since the service curve $\gamma_S (\kappa) = \tfrac{\kappa}{R}$ does not allow a delay 
at the start of a busy period, we get $B(T_A (\underline{\nu})) =  W(\underline{\nu}) = 0$. )
With $T_D (\kappa) \le T_A \maxconv \gamma_S (\kappa)$ for all $\kappa$, we therefore have 
\begin{align}
T_A (\kappa) < T_A \maxconv \gamma_S (\kappa) \, , \ \ \forall \kappa \in (\underline{\nu},  \nu ] \, .
\label{eq:TAlessthanconv} 
\end{align}
Under these assumptions, the interval over which $T_A \maxconv \gamma_S (\nu)$ is computed can be 
reduced as given in the following lemma. 
\begin{lemma}
\label{lemma:last}
Given an arrival $\nu$ at time $t$ to a network element that offers the service curve 
$\gamma_S (\nu) = \tfrac{\nu}{R}$. If $\nu > \underline{\nu}$, 
then 
\[
T_A \maxconv \gamma_S (\nu) = \sup_{0 \le \kappa \le \underline{\nu}} \{ T_A (\kappa) + \frac{\nu - \kappa}{R} \} \, . 
\]
\end{lemma}
\begin{proof}
Since $T_A$ and $\gamma_S$ are right-continuous, by \cite[Lemma~4.1(9)]{Duality}, there exists a $\mu \in [0, \nu]$ such that 
\[
T_A \maxconv \gamma_S (\nu) = T_A (\mu)  + \frac{\nu - \mu}{R} \, . 
\]
If $\mu \in (\underline{\nu},\nu ]$, we get  
\begin{align*}
T_A \maxconv \gamma_S (\nu) & = T_A (\mu)  + \frac{\nu - \mu}{R} \\
& < \sup_{0 \le \kappa \le \mu}   \{ T_A (\kappa) + \frac{\mu - \kappa}{R} \} + \frac{\nu - \mu}{R} \\
& \le  \sup_{0 \le \kappa \le \nu}   \{ T_A (\kappa) + \frac{\nu - \kappa}{R} \}  \\
& = T_A \maxconv \gamma_S (\nu) \, , 
\end{align*}
In the second line, we used~\eqref{eq:TAlessthanconv}, and in the third line, we enlarged the 
range of the supremum. Obviously, there is a contradiction, and we can conclude that $\mu \le \underline{\nu}$. 
\end{proof}

Rewriting the result in Lemma~\ref{lemma:last} as 
\begin{align*}
T_A \maxconv \gamma_S (\nu) & = \sup_{0 \le \kappa \le \underline{\nu}}   \{ T_A (\kappa) + \frac{\underline{\nu} - \kappa}{R} \}
+  \frac{\nu - \underline{\nu}}{R} \\
& \hspace{-4mm} = \max \bigl[ \sup_{0 \le \kappa < \underline{\nu}} \{ T_A (\kappa) + \frac{\underline{\nu} - \kappa}{R} \}
, T_A (\underline{\nu}) \bigr] +  \frac{\nu - \underline{\nu}}{R} \\
& \hspace{-4mm}  = \max\{T_A \maxconv \gamma_S (\underline{\nu}^-) , \underline{t}) \} +  \frac{\nu - \underline{\nu}}{R} \, , 
\end{align*}
where we use the notation $x^- = \sup_{y < x} y$, 
we can construct a deadline assignment for 
the rate service curve. 
Let us add an index to the  busy periods so that $\underline{t}_n$ and $\underline{\nu}_n$ denote the start time and the first bit of the $n$th busy period of a flow. 
Then the deadline assignment in the $n$th busy period is given by 
\begin{align}
& \Dl (\nu) = \max \left( \Dl  (\underline{\nu}_n^-) , \underline{t}_n \right)
+ \frac{\nu - \underline{\nu}_n}{R} , \ \nu \in [\underline{\nu}_n,  \underline{\nu}_{n+1} ) \, . 
\label{eq:deadline-rate}
\end{align}
For the computation of the first busy period, we define $ \Dl  (\underline{\nu}_1^-) = -\infty$. This is consistent with our derivations since 
$\underline{\nu}_1= 0$ and, with $T_A \maxconv \gamma_S \in \T_o$, we get $T_A \maxconv \gamma_S (\nu)= -\infty$ for $\nu < 0$. 

This deadline assignment is easily implemented, since we must only keep track 
of the arrived bits in the current busy period. Consider the $n$th busy period which starts at 
$ \underline{t}_n$ with bit value $\underline{\nu}_n$. 
At the begin of a busy period, we take the larger of the current time ($ \underline{t}_n $) 
and the last assigned deadline ($\Dl  (\underline{\nu}_n^-)$). This value is added to $\frac{\mu}{R}$ to obtain the deadline of the $\mu$th bit 
in the busy period. 
Since $\nu$ in \eqref{eq:deadline-rate} is equal to $\nu = \mu + \underline{\nu}_n$ for 
$\nu \in [\underline{\nu}_n,  \underline{\nu}_{n+1} )$, the resulting deadline is equal to \eqref{eq:deadline-rate}.
Note that the start time of a busy period is simply the time of an arrival to an empty buffer.



{\bf Latency-Rate Guarantees:} 
We can combine the deadline assignment of  a delay server and  
a rate server to get the deadline assignment of a
latency-rate server. Let $\gamma_1(\nu) = \tfrac{\nu}{R}$ and  
$\gamma_2(\nu) = d$, the service curve of a latency-rate server is 
$\gamma_1 \maxconv \gamma_2 (\nu) =  \tfrac{\nu}{R} +d$.  
Since 
\[
T_A \maxconv ( \gamma_1 \maxconv \gamma_2 )  (\nu) = T_A \maxconv \gamma_1 \ (\nu) + d \, ,
\]
the deadline of $\nu$ for a latency-rate server is given by  
$\overline{\Dl} (\nu)$ from 
\begin{align}
\overline{\Dl} (\nu) = \Dl (\nu) + d \, , 
\label{eq:deadline-latency-rate}
\end{align}
where $\Dl (\nu)$ is the deadline  for $\gamma_1$ computed with~\eqref{eq:deadline-rate}. 


\balance
{\bf Piecewise linear convex service curve:}  The deadline computation with \eqref{eq:deadline-rate} and~\eqref{eq:deadline-latency-rate} can  be extended to piecewise linear convex max-plus service curves. A single segment of such a service curve has the form 
\[
\gamma_S (\nu) = \left[ \frac{\nu}{r} - e \right]^+
\]
for some  $r > 0$ and $e \ge 0$. 
For $e = 0$, this service curve is obviously a rate server, and the deadline computation from~\eqref{eq:deadline-rate} applies. For $e > 0$, we essentially have a delay correction with a negative value. Since the earliest deadline of a packet is its arrival time, we compute the deadline as 
\begin{align}
\overline{\Dl} (\nu) = \max \{ \Dl (\nu) -e , T_A (\nu) \} \, . 
\label{eq:deadline-convex-segment}
\end{align}
with ${\Dl}(\nu) $ from~\eqref{eq:deadline-rate}. We obtain a piecewise convex max-plus service curve with 
multiple segments, from 
\[
\gamma_S (\nu) = \max_{i=1, \ldots, N} \left\{ \left[ \frac{\nu}{r_i} - e_i \right]^+\right\} \, , 
\]
with 
$e_1 < e_2 < \ldots < e_N$ and $R_1 < R_2 < \ldots < R_N$. 
The deadline for the piecewise linear convex max-plus service curve 
is computed by $\overline{\Dl} (\nu)  = \max_{i=1, \ldots, N} \overline{\Dl}_i (\nu)$, where $\overline{\Dl}_i$ is the deadline computed for the $i$th segment. 

{\bf Traffic shaping:} The SCED principle is also applicable to  traffic shaping. A max-plus traffic envelope $\lambda_E \in \T_o$ realizes an exact service curve, with $T_D (\nu) = T_A \maxconv  \lambda_E(\nu)$. Here, the convolution  provides the time when the shaper releases bit $\nu$. 
We therefore refer to the max-plus convolution as 
the release time and denote it by $r\ell$, with  
\[
r\ell (\nu) =  T_A \maxconv  \lambda_E (\nu) \, . 
\] 
As an example, the max-plus envelope for a token bucket with 
rate $r$ and bucket size $b$ has the envelope 
$\lambda (\nu) = [ \tfrac{\nu}{r} - \tfrac{b}{r}]^+$. 
We compute the release times with~\eqref{eq:deadline-rate} and~\eqref{eq:deadline-convex-segment}, were we replace `$\Dl$' by `$r\ell$'.  

\section{Packetized Systems}
\label{sec:packet}

While a fluid-flow interpretation of SCED is perfectly aligned 
with network calculus theory, verifying and enforcing deadlines  
for each (real) value $\nu$ of a traffic flow is obviously not 
practical. In a packet system, 
each packet is assigned a single deadline, and all bits 
belonging to the same packet receive the same deadline. 
We now discuss  adjustments of SCED for  a packet-level system. 

Let $\ell_n$ denote the size of the $n$th packet ($n \ge 1$) of a flow  
and $L_n = \sum_{k=1}^N {\ell_n}$ the cumulative size of the first $n$ packets, 
with $L_o = 0$. The bits of the $n$th packet 
cover the range  $L_{n-1} \le \nu < L_n$. 
For a system with packet-level arrivals we have 
\[
T_A (\nu) = T_A^p (n) \, , \qquad \nu \in [L_{n-1 }, L_n) \, , 
\]
where $T_A^p (n)$ is the arrival time of packet $n$. 
With a  service curve for delays, $\gamma_S (\nu)=d$, the (fluid-flow) SCED deadline 
assignment according to~\eqref{eq:Deadline-assign} is 
\[
\Dl (\nu) = T_A^p (n) + d \, ,  \qquad \nu \in [L_{n-1 }, L_n) \, ,
\] 
hence a packet-level deadline assignment $\Dl^p (n) = T_A^p (n)+ d$ is congruent 
with the fluid-flow  assignment. 

With a rate-guarantee with $\gamma_S (\nu) = \tfrac{\nu}{R}$, fluid-flow SCED assigns each bit value of a packet 
a different deadline. For a packet-level system, we use \cite[Eq.~(12.14)]{Duality} which showed  
\[
T_A^p \maxconv \gamma_S (\nu) = \max \bigl\{ T_A \maxconv \gamma_S (L_{k-1}^-) ,  
T^p_A (k)  \bigr\} + \frac{\nu-L_{k-1}}{R} \, ,   
\]
for $\nu \in [L_{n-1 }, L_n)$.
Hence, a packet-level deadline assignment 
\begin{align}
\Dl^p (n) = \max \bigl\{ \Dl^p (n-1) ,  T^p_A (k)  \bigr\} + \frac{\ell_n}{R} \, ,   
\label{eq:VC-deadline}
\end{align}
relates to the fluid-flow assignment $\Dl (\nu)$ by 
\[
\Dl (\nu) \le \Dl^p (n) - \tfrac{\ell_n}{R}  \, , \qquad \nu \in [L_{n-1 }, L_n) \, .
\]
Therefore, by adjusting the service curve by an additional delay, yielding
$\gamma_{S'}(\nu) = \tfrac{\nu+ \ell_{\max}}{R}$, where $\ell_{\max}$ is the maximum packet size of the flow, 
the deadline $\Dl' = T_A \maxconv \gamma_{S'}$ 
satisfies  
\[
\Dl' (\nu) \le \Dl^p (n)   \, , \qquad \text{if } \nu \in [L_{n-1 }, L_n) \, .
\]
We see that the packet-level assignment in~\eqref{eq:VC-deadline} meets all deadlines of a 
fluid-flow assignment with the adjusted rate service curve~$\gamma_{S'}$. The deadline assignment in~\eqref{eq:VC-deadline} is of course that of the 
VirtualClock scheduling algorithm in~\cite{VirtualClock90}.

\section{Conclusions}
By resorting to max-plus algebra for the computation of deadlines, 
we showed that a SCED scheduler can be much simplified. 
The computation of deadlines for latency-rate and piecewise-linear convex   
max-plus service curves only requires state information on the start time and the traffic served in the current busy period. A packet-level algorithm for deadline computations 
emphasized the relationship of SCED for  rate-based service curves and 
 VirtualClock scheduling. 

\section*{Acknowledgements}
This work is supported in part by the Natural Sciences and Engineering Research Council of Canada (NSERC). The author would like to thank Almut Burchard and Natchanon Luangsomboon for their discussions and feedback.

\bibliographystyle{plain}

\end{document}